\documentclass[a4paper,11pt,reqno]{amsart}
\usepackage{amssymb}
\usepackage{enumitem}
\usepackage{amsthm}
\usepackage[alphabetic]{amsrefs}
\usepackage{tikz}
\usepackage{blindtext, xcolor}
\usepackage{comment}
\usepackage{mathrsfs}
\usepackage{a4wide}
\usepackage{fullpage}
\usepackage{mathtools}
\usepackage{hyperref}

\usepackage{algorithm}
\usepackage{algpseudocode}

\usepackage{bm}

\newtheorem{corollary}{Corollary}[section]
\newtheorem{theorem}[corollary]{Theorem}
\newtheorem{lemma}[corollary]{Lemma}
\newtheorem{proposition}[corollary]{Proposition}

\newtheorem{remark}[corollary]{Remark}

\newtheorem{assumption}{Assumption}

\numberwithin{equation}{section}

\newcommand{\feri}[1]{{\bf \textcolor{olive}{#1 - F} }}

\newcommand{\st}[1]{{\star_{#1}}}

\usepackage{graphicx} 

\title{Barvinok's interpolation method meets Weitz's correlation decay approach}

 \author{Ferenc Bencs}
 \email{\texttt{ferenc.bencs@gmail.com}}
  \address{Centrum Wiskunde \& Informatica, P.O.Box 94079
 1090 GB Amsterdam, The Netherlands.}
 \author{Guus Regts}
 \email{\texttt{guusregts@gmail.com}}
 \address{Korteweg de Vries Institute for Mathematics, University of Amsterdam, P.O. Box 94248  1090 GE Amsterdam, The Netherlands.}
 \thanks{FB is funded by the Netherlands Organisation of Scientific Research (NWO): VI.Veni.222.303. GR is funded by the Netherlands Organisation of Scientific Research (NWO): VI.Vidi.193.068}

\date{\today}

\begin{document}

\begin{abstract}
In this paper we take inspiration from Weit'z algorithm for approximating the independence polynomial to provide a new algorithm for computing the coefficients of the Taylor series of the logarithm of the independence polynomial.
Hereby we provide a clear connections between Barvinok's interpolation method and Weitz's algorithm.
Our algorithm easily extends to other graph polynomials and partition functions and we illustrate this by applying it to the chromatic polynomial and to the graph homomorphism partition function.  
Our approach arguably yields a simpler and more transparent algorithm than the algorithm of Patel and the second author.

As an application of our algorithmic approach we moreover derive, using the interpolation method, a deterministic $O(n(m/\varepsilon)^{7})$-time algorithm that on input of an $n$-vertex and $m$-edge graph of minimum degree at least $3$ and $\varepsilon>0$ approximately computes the number of sink-free orientations of $G$ up to a multiplicative $\exp(\varepsilon)$ factor.
\end{abstract}

\maketitle

\section{Introduction}
Weitz's correlation decay approach\footnote{It should be noted that while this approach is often attributed to Weitz~\cite{Weitz}, also Bandyopadhyay and Gamarnik~\cite{BanGar} developed algorithms for approximating partition function based on the notion of correlation decay around the same time as Weitz's breakthrough result.}~\cite{Weitz} and Barvinok's interpolation method~\cites{Barbook,PatReg17} are both algorithmic techniques to design efficient deterministic approximation algorithms for counting problems such as counting the number of independent sets in a graph or counting the number of proper $q$-colorings of a graph. These counting problems can often be cast as the evaluation of partition functions of statistical physics models such as the hard-core model, and the Potts model.
While other techniques have been developed to design efficient algorithms for these tasks over the years~\cites{Moitra,HPR21,chen2025deterministiccountingcouplingindependence}, the correlation decay approach and the interpolation method appear to have been the most widely used ones.

At first sight, Weitz's correlation decay approach and Barvinok's interpolation method appear to be very different in nature, but surprisingly they yield very similar results for a variety of models such as the hard-core model~\cites{Weitz,PRSokal}, the matching polynomial~\cites{Bayatietalmatchings,PatReg17}, the edge cover polynomial~\cites{edge,weighted,bencs2020some} and the graph homomorphism partition function~\cites{lu2013improved,BarvinokSoberon,Barbook,PatReg17}. 
The main contribution of the present paper is to bridge these two algorithmic approaches. 
To discuss our contribution we will now specialize to the situation of approximating the partition function of the hard-core model, which is also the model for which Weitz originally invented his approach.

\subsection{The hard-core model}
Consider the partition function of the hard-core model of a graph $G=(V,E)$, which is defined as 
\begin{equation}\label{eq:def Z_G}
Z(G;\lambda)=\sum_{I\in \mathcal{I}_G}\lambda^I,
\end{equation}
where $\mathcal{I}_G$ denotes the collection of all \emph{independent sets} of $G$, i.e. subsets of the vertices that do not span any edges.
The partition function is also known as the \emph{independence polynomial} of $G$.
About twenty years ago Weitz~\cite{Weitz} developed a polynomial time algorithm to compute $Z_{G}(\lambda)$ within a multiplicative $\exp(\varepsilon)$ factor for graphs of a given maximum degree $\Delta$ provided $\lambda<\lambda_c(\Delta):=\frac{(\Delta-1)^{\Delta-1}}{(\Delta-2)^\Delta}$. (Later it was shown that when $\lambda>\lambda_c(\Delta)$ it is in fact \textsc{NP-hard} to approximate $Z_G(\lambda)$, see~\cites{BeyondlambdacSlyandSun,BeyondlambdacGalanisetal}.)
\subsubsection*{Weitz's algorithm}
Instead of approximating $Z(G;\lambda)$ directly, Weitz~\cite{Weitz} approximates \emph{ratios} of the form
\begin{align}\label{eq: ratio}
   \frac{Z^{v \text{ out}}(G;\lambda)}{Z(G;\lambda)},
\end{align}
where $v$ is a vertex of $G$ and by $Z^{v \text{ out}}(G;\lambda)$ (resp. $Z^{v \text{ in}}(G;\lambda)$) we denote the summation~\eqref{eq:def Z_G} restricted to those independent sets $I\in \mathcal{I}_G$ that do not contain the vertex $v$ (resp. that \emph{do} contain the vertex $v$).
For positive $\lambda$ the ratio~\eqref{eq: ratio} has the interpretation that the vertex $v$ is not in the random independent set ${\bf I}$ drawn proportionally to $\lambda^{|I|}$.
Note that $Z^{v \text{ out}}(G;\lambda)=Z_{G-v}(\lambda)$ and $Z^{v \text{ in}}(G;\lambda)=\lambda Z(G\setminus N[v];\lambda)$.
Weitz then combines these ratios into a telescoping product to approximate 
\begin{align}\label{eq:tel prod}
\frac{1}{Z(G;\lambda)}=\frac{Z(G-v_1;\lambda)}{Z(G;\lambda)}\frac{Z(G-v_1-v_2;\lambda)}{Z(G-v_1;\lambda)} \cdots \frac{Z(G\setminus V_{n-1}-v_n;\lambda)}{Z(G\setminus V_{n-1};\lambda)},
\end{align}
where we fix some ordering of the vertex set $V=\{v_1,\ldots,v_n\}$ and $V_{i}$ denotes the set $\{v_{1},\ldots,v_i\}$.
To approximately compute the ratio~\eqref{eq: ratio} Weitz actually considers the ratio 
\[
R_{G,v}(\lambda):=\frac{Z^{v \text{ in}}(G;\lambda)}{Z^{v \text{ out}}(G;\lambda)},
\]
which is just a simple transformation of~\eqref{eq: ratio}.
He iteratively expresses $R_{G,v}(\lambda)$ in terms of ratios of smaller graphs and truncates this process at suitable depth (logarithmic in the number of vertices),  thereby obtaining a polynomial time algorithm provided the maximum degree of the graph is bounded.
The result is an $\exp(\varepsilon)$ approximation to the partition function $Z(G;\lambda)$ at $\lambda$, provided the hard-core model satisfies what is called \emph{strong spatial mixing} (with exponential decay), which he proves to be the case whenever $\lambda<\lambda_c(\Delta)$ is fixed in advance.

\subsubsection*{Barvinok's interpolation method}
Barvinok's interpolation method for the partition function of the hardcore model/the independence polynomial roughly works as follows.
To approximate the polynomial $Z(G;\lambda)$ at some $0<\lambda<\lambda_c(\Delta)$, on the family, $\mathcal{G}_\Delta$, of graphs of maximum degree at most $\Delta$, one needs an open set $U\subset\mathbb{C}$ containing $0$ and $\lambda$ such that for all $x\in U$ and $G\in \mathcal{G}_\Delta$ $Z(G;x)\neq 0.$ 
The existence of such a set is guaranteed by~\cite{PRSokal}. 
One then defines $f(x)=\log(Z(G;x))$ on $U$ (where we fix the principal branch of the logarithm) and uses the first logarithmically many (in the number of vertices) coefficients of the Taylor series of $f(x)$ to obtain a good approximation to $Z_G(\lambda)$. 
In case $U$ is a disk centered at zero of radius larger than $\lambda$, this follows from~\cite{Barbook}*{Lemma 2.2.1}, while if $U$ is of a different shape, for example a rectangle containing a real interval, this requires some additional work, see for example~\cite{Barbook}*{Section 2.2.3} or~\cite{PatReg17}*{Section 4.3}. 
Computing the coefficients of $\log(Z_G(x))$ in a brute force manner leads to quasi-polynomial time algorithms. This has been improved to genuine polynomial time with a more refined algorithm in~\cite{PatReg17}.

\subsubsection*{A comparison}
Both algorithmic approaches consist of two main ingredients. 
First of all,  in both cases there is an \textit{algorithm} to compute a number and secondly a \textit{condition} that guarantees that the output of this algorithm is a good approximation to the evaluation of the partition function, strong spatial mixing for the correlation decay approach and a suitable zero-free region for the interpolation method.

At first sight these two algorithmic approaches may appear very different for \textit{both} ingredients.
However, recent developments have shown that the second ingredient for the respective methods (strong spatial mixing, resp. zero-freeness) are in fact closely related~\cites{LSSfisher,LSS2Delta,Shaounified,gamarnik2022correlation,regtsabsence,shao2024zero}. 

The goal of the present paper is to show that also the other ingredient, the actual algorithms, are in fact closely related.
We do this by providing an alternative, and in our opinion much simpler algorithm than the algorithm from~\cite{PatReg17}, for computing the coefficients of the Taylor series in Barvinok's interpolation method, taking inspiration from Weitz's algorithm.

\subsubsection*{Our algorithm}
We describe the main ideas of the algorithm for the partition function of the hard-core model and indicate how this bridges the two approaches. 
We note here that these ideas, that are partly inspired by~\cite{HPR21}, may likely be known to some of the experts in the area, but as far as we know they have never appeared in print.

We want to compute the first, say $\ell$, coefficients of the Taylor series of 
\begin{align}\label{eq:log Z}
f(x)=\log(Z(G;x))=\sum_{k\geq 1} c_k(G)x^k.
\end{align}
The idea is to take the derivative at both sides of~\eqref{eq:log Z} arriving at 
\[
x\frac{\tfrac{d}{d x}Z(G;x)}{Z(G;x)}=\sum_{k\geq 1} kc_k(G)x^{k}
\]
and compute the first $\ell$ coefficients of this series, from which the original coefficients can easily be deduced.
Now note that
\begin{equation}\label{eq:derivative log=sum ratios}
x\tfrac{d}{dx}Z(G;x)=\sum_{I\in \mathcal{I}_G}|I|x^{|I|}=\sum_{v\in V}\sum_{I\in \mathcal{I}_G: v\in I}x^{|I|}=\sum_{v\in V}Z^{v \text{ in}}(G;x).
\end{equation}
Consequently, 
\[
x\frac{\tfrac{d}{dx}Z(G;x)}{Z(G;x)}=\sum_{v\in V}\frac{Z^{v \text{ in}}(G;x)}{{Z(G;x)}}=\sum_{v\in V}\frac{Z^{v \text{ in}}(G;x)}{{Z^{v\text{ in}}(G;x)+Z^{v\text{ out}}(G;x)}}=\sum_{v\in V}\frac{R_{G,v}(x)}{1+R_{G,v}(x)}.
\]
and it thus suffices to compute the first $\ell$ coefficients of the series expansion of $R_{G,v}(x)$ for each $v\in V$.
We do this iteratively by expanding the ratio $R_{G,v}(x)$ in terms of ratios of smaller (induced subgraphs) just as Weitz does in his algorithm~\cite{Weitz}. 
In our opinion this provides a clear bridge between the algorithmic parts of the two approaches\footnote{We note that some time after the initial posting of the present paper to the arXiv, Shao and Shi~\cite{shao2025zero} showed a way of using zero-freeness in combination with essentially the first ingredient of Weitz's approach to design efficient approximation algorithms for evaluating partition functions, that relies on similar ideas as those that have been developed here.}.

Some remarks are in order. First of all while Weitz only cares about the numerical value of the ratio, we care about the coefficients of its series expansion.
Secondly, to approximate the partition function, Weitz has to (approximately) compute each ratio in the telescoping product~\eqref{eq:tel prod}, while we have to compute the coefficients of the ratio $R_{G,v}(x)$ for each vertex $v$ of $G$. 
A possible advantage of that is that it is easier to take advantage of the presence of symmetry. For example in case the graph $G$ is vertex transitive all ratios are equal and we only have to compute one, while after removing vertices, as is done in~\eqref{eq:tel prod}, one may loose this property.

Additionally, it is easy to parallelize our algorithm in a naive fashion, as already we can work for each $v\in V$ with the ratios $R_{G,v}(x)$ independently. We will see that the ratios $R_{G,v}(x)$ themselves also carry a natural recursion, which is easy to parallelize.

\subsection{Our contributions}
Our approach for computing the coefficients of the Taylor expansion of $\log(Z(G;x))$ is, in fact, applicable to many other models.
However, rather than setting up a general framework for which our algorithmic approach applies, we will only describe it for several concrete models of increasing complexity, from which the applicability of the approach should become clear.
In the following subsections we will respectively describe our algorithm for computing the relevant Taylor coefficients for the hard-core model, the number of sink-free orientations, the chromatic polynomial and counting graph homomorphisms.
Combining this with the guarantee for a good approximation for Barvinok's method (i.e. a suitable zero-free region) our algorithm yields a fully polynomial time approximation scheme (FPTAS) for these models/polynomials. 
The resulting FPTAS is not presented for the models we consider, we only explicitly establish it for counting sink-free orientations.

In all our algorithm we will work with the following assumption.
\begin{assumption}\label{as:computation}
In our algorithms we assume always that the graph $G$ is given in adjacency list format, such that the algorithm can determine for any vertex $u\in V$ the neighbors of $u$ in $G$ in constant time per neighbor.
\end{assumption}
We will hide in the big $O$ notation the constants that do not depend on $\Delta$, which will denote the maximum degree of the graph.
For a power series $f(x)=\sum_{k\geq 0}a_k x^k$ and $m\in \mathbb{N}_{\ge 1}$ we denote by $f(x)^{[m]}:=\sum_{k=0}^{m} a_k x^k$ the $m$th Taylor polynomial of $f(x)$.

\subsection*{Hard-core model}
In Section~\ref{sec:hard-core} we work out the details for the hard-core model that we sketched in the previous subsection. 
This is the simplest possible setting and clearly shows the basic features of the algorithm. 

We state here the precise algorithmic statement that we obtain.
\begin{theorem}\label{thm:hard-core}
Under Assumption~\ref{as:computation}, there exists an $O(n\Delta^{m}m^3)$-time algorithm that on input of an $n$-vertex graph $G$ of maximum degree at most $\Delta$ and $m\in \mathbb{N}$ that computes  $\log(Z(G;x))^{[m]}$.
\end{theorem}

\subsection*{Sink free orientations}
In Section~\ref{sec:sink free} we consider sink-free orientations. 
Recall that an orientation of the edges of a graph $G$ is called \emph{sink free} if each vertex in the resulting directed graphs has at least one outgoing edge.
Very recently Anand et. al.~\cite{anand2025sinkfreeorientationslocalsampler} gave a deterministic algorithm that on input of an $n$-vertex graph of minimum degree at least $3$ and $\varepsilon>0$ computes the number of sink-free orientations of $G$ with in a multiplicative $\exp(\varepsilon)$ factor in time $O((n/\varepsilon)^{72}n\log(n/\varepsilon))$.

As observed by Anand et. al.~\cite{anand2025sinkfreeorientationslocalsampler}, it follows from the inclusion-exclusion principal that we can express the number of sink-free orientations as a (multivariate) evaluation of the independence polynomial. 
For completeness we provide a short proof below.
\begin{lemma}\label{lem:sink free}
For any graph $G$ with $m$ edges, 
\begin{equation}\label{eq:sink free}
    \textrm{sfo}(G)=2^{m}\sum_{S\in\mathcal{I}(G)}\prod_{v\in S}(-(1/2)^{\deg_G(v)}).
\end{equation}
\end{lemma}
As mentioned in~\cite{anand2025sinkfreeorientationslocalsampler} one cannot directly invoke the interpolation method, since a naive application of it would result in quasi-polynomial running time.
We therefore introduce the following univariate polynomial for a graph $G=(V,E)$
\begin{equation}
Z_{\textrm{sfo}}(G;t):=\sum_{S\in\mathcal{I}(G)}\prod_{v\in S}(-t^{\deg_G(v)}),
\end{equation}
and realize that $Z_{\textrm{sfo}}(G;1/2)=2^{-|E|}\textrm{sfo(G)}$.
We prove in Section~\ref{sec:sink free} a suitable zero-free region for this polynomial, providing the first ingredient for Barvinok's interpolation method.
For the second ingredient we extend our algorithm for computing Taylor coefficients of $\log(Z(G;x))$ to $\log(Z_{\textrm{sfo}}(G;t))$ and thereby we obtain the following result:
\begin{theorem}\label{thm:sink free}
There exists a deterministic algorithm that on input of a graph $G$ with $n$ vertices and $m$ edges of minimum degree at least $3$ and $\varepsilon>0$, computes the number of sink-free orientations of $G$ with in a multiplicative $\exp(\varepsilon)$ factor in time $O(n(m/\varepsilon)^{7})$. 
\end{theorem}

\begin{remark}
We have to leave it as an interesting open question to devise an efficient algorithm for graphs of minimum degree at least $2$. We note that zeros of the independence polynomial of cycles approach $-1/4$~\cites{ScottSokal,deboerratios}, or in other words zeros of $Z_{\textrm{sfo}}(C_n;t))$ approach $1/2$ (here $C_n$ denotes the cycle of length $n$).
This implies that the interpolation method cannot be applied in a naive manner in case the graph has vertices of degree two.
We note that~\cite{anand2025sinkfreeorientationslocalsampler} cannot handle graphs with degree two vertices for similar reasons (because the presence of zeros essentially implies that there is no constant lower bound on the marginal probability that a given vertex is not a sink).
\end{remark}

Before moving to the chromatic polynomial we provide here a proof of Lemma~\ref{lem:sink free}.
\begin{proof}[Proof of Lemma~\ref{lem:sink free}]
The identity follows from the inclusion exclusion principle realizing that when randomly orienting each edge with probability $1/2$, the probability that the constraint at a vertex is violated is exactly $(1/2)^{\deg_G(v)}$. 
Bad events correspond exactly to independent sets (since any neighbor of a vertex that is a sink is automatically not a sink).
Therefore, the right-hand side of~\eqref{eq:sink free} is exactly $2^{m}$ times the probability that such an orientation is sink free.
\end{proof}

\subsection*{The chromatic polynomial}
For a graph $G=(V,E)$ let $\chi(G;q)$ denote its chromatic polynomial, the unique monic polynomial of degree $|V|$ that satisfies that for each $q\in \mathbb{N}$, $\chi(G;q)$ is equal to the number of proper $q$-colorings of $G.$
In the context of Barvinok's interpolation method, relevant zero-free regions for the chromatic polynomial are of the form $\chi(G;q)\neq 0$ provided $|q|\geq K\Delta(G)$ (where $\Delta(G)$ denotes the maximum degree of $G$) for a uniform constant $K$. 
The first such bound was proved by Sokal~\cite{Sokalzeros} with various improvements over the years~\cites{FerProc,JPR2024,bencs2025improvedboundszeroschromatic}. 
The currently best known constant is $K=4.25$ and is due to the authors of the present paper~\cite{bencs2025improvedboundszeroschromatic}.

To apply the interpolation method we have to compute the low-order Taylor coefficients of the logarithm of the polynomial 
\[
P(G;z):=(-z)^{|V|}\chi(G;-1/z),
\]
which has no zeros in the disk centered at zero of radius $\tfrac{1}{4.25 \Delta(G)}$ by~\cite{bencs2025improvedboundszeroschromatic}.
In~\cite{PatReg17} this is done expressing the coefficients of $P(G;z)$ in terms of induced subgraph counts. 
Here we will take advantage of the fact that we can view $P(G;z)$ as the generating function of so-called \emph{broken circuit-free} forests by Whitney's theorem. 
One could thus view $P(G;z)$ as an evaluation of the multivariate independence polynomial of an associated graph. 
We then build on and extend the algorithm for computing coefficients of $\log(Z(G;x))$ and use it to design a simple algorithm to compute the coefficients of $\log(P(G;z))$ summarized in the following theorem.
\begin{theorem}\label{thm:chrom}
There exists an algorithm that on input of an $n$-vertex graph of maximum degree at most $\Delta$ and $m\in \mathbb{N}$ computes $\log(P(G;z))^{[m]}$ in $O(n\Delta(3e\Delta)^mm^3)$ time.
\end{theorem}

\subsection*{Graph homomorphisms}
Let $q\geq 2$ be an integer and let $A$ be a symmetric $q\times q$ matrix. 
For a graph $G=(V,E)$ let $\hom(G,A)$ denote the "number" of homomorphisms of $G$ into $A$, that is,
\[
    \hom(G,A)=\sum_{\phi:V\to[q]}\prod_{uv\in E}A_{\phi(u)\phi(v)}.
\]
In particular, if $A=A(H)$ is the adjacency matrix of some graph $H$, then $\hom(G,A)$ is exactly the number of homomorphism numbers of $G$ to $H$. 
In particular, if $H$ is the complete graph on $q$ vertices, $K_q$, then $\hom(G,A(K_q))=\chi(G,q)$. 
One could think about graph homomorphisms as multi-spin systems.

To apply Barvinok's interpolation method, let us encode the number of homomorphisms into a polynomial evaluation, by defining
\[
H(G;x):=q^{-|V|}\hom(G,J+x(A-J)),
\]
here $J$ denotes the $q\times q$ all ones matrix.
By definition, we see that $H(G;1)=q^{-|V|}\hom(G,A)$. 
To obtain an efficient algorithm for approximating $H(G;A)$ for the family of graphs of maximum degree at most $\Delta$ based on the interpolation method one requires an open set $\mathcal{U}\subset \mathbb{C}$ that contains the points $x=0$ and $x=1$ such 
that $H(G;x)\neq 0$ for all $x\in \mathcal{U}$ and graphs $G$ of maximum degree at most $\Delta$. See~\cite{Barbook}*{Section 2.2} and \cite{PatReg17}.
For example, this is  provided if $A\in\mathbb{C}^{q\times q}$ is close enough to the matrix $J$ in the infinity norm~\cite{Barbook} and also if $A=A(K_q)$ provided $q\geq 1.998\Delta$~\cites{bencs2024deterministicapproximatecountingcolorings,LSScorrelation}.
Having such a zero-free region one then needs to compute the first $O(\log|V|/\varepsilon)$ coefficients of $\log(H(G,x))$. 
In~\cite{PatReg17} this is done expressing these coefficients as induced graph counts.

In Section~\ref{sec:graph hom} we will devise a simple algorithm based on our algorithm for the chromatic polynomial.
The main idea is to realize $H(G;x)$ as a multivariate evaluation of the forest generating function following~\cite{bencs2024approximatingvolumetruncatedrelaxation}. 
After this is done, the algorithm proceeds much in the same way as for the chromatic polynomial.
We record here the relevant statement.
\begin{theorem}\label{thm: graph hom}
Let $q\in \mathbb{N}_{\geq 2}$ and let $A$ be symmetric $q\times q$ matrix.
Then there exists an algorithm that on input of a graph $G=(V,E)$ of maximum degree at most $\Delta$ and $m\in \mathbb{N}$ computes $\left(\log(H(G;x))\right)^{[m]}$ in time $O((e(q+1)\Delta)^mm^3)$.
\end{theorem}

\subsection*{Organization}
The next section contains some relevant preliminaries. The remaining sections are devoted to the independence polynomial, sink-free orientations, the chromatic polynomial and graph homomorphisms respectively. In Section~\ref{sec:hard-core} we record the relevant algorithms for the independence polynomial in pseudo code, while in the other sections we have chosen not to do this to save space since the pseudo code for the independence polynomial can easily be adapted to these other settings.

\section{Preliminaries}\label{sec:prel}
In this section we collect some facts, definitions and notation about power series that will be used throughout later.

Let $U\subseteq \mathbb{C}$ be an open set containing $0$.
For an analytic function $f:U\to \mathbb{C}$ with Taylor-series at $0$, i.e.
\[
f(z)=\sum_{k\ge 0} a_kz^k,
\]
we let $f^{[m]}(z)$ denote the $m$th Taylor polynomial of $f$, i.e. 
\[
f^{[m]}(z)=\sum_{k=0}^{m} a_kz^k.
\]
Now let $\st{m}$ denote polynomial multiplication up to the $m$th coefficient, i.e.
\[
p(z)\st{m}q(z)=\left(p(z)q(z)\right)^{[m]}.
\]
Note that, $p(z)\st{m}q(z)=p^{[m]}(z)\st{m}q^{[m]}(z)$, which  means that to compute $p(z)\st{m}q(z)$ we need at most $O(m^2)$ many operations. (Using fast polynomial multiplication one could obtain a better running time.)

\begin{lemma}\label{lemma:reciprocial}
Let $f:U\to \mathbb{C}$ be an analytic function with a Taylor-series at $0$ such that $f(0)=0$. Then there is an $O(m^3)$-time algorithm that on input of $f^{[m]}(z)$ computes
\[
\left(\frac{1}{1-f(z)}\right)^{[m]}=\left(\frac{1}{1-f^{[m]}(z)}\right)^{[m]}.
\]
\end{lemma}
\begin{proof}
    By definition we know that there is an analytic function $g(z)$ such $f(z)=zg(z)$. 
    The Taylor series at $z=0$ of $1/(1-z)$ is given by $\sum_{k\ge 0} z^k$. Thus
    \[
        \frac{1}{1-f(z)}=\sum_{k\ge 0} f(z)^k=\sum_{k\ge 0} z^kg(z)^k.
    \]
    By taking the terms of degree at most $z^{m}$ we get that
    \[
    \left(\frac{1}{1-f(z)}\right)^{[m]}=\left(\sum_{k=0}^{m}z^kg^{[m]}(z)^k\right)^{[m]}=\underbrace{(\dots((}_{m}zg^{[m]}(z)+1)\st{m}zg^{[m]}(z)+1)\st{m}\dots+1)^{[m]},
    \]
    which proves the desired statement.
\end{proof}

\section{The hard-core model}\label{sec:hard-core}
Recall from the introduction that for a graph $G=(V,E)$ we denote by $Z(G;x)$ the partition function of the hard-core model on $G$. The polynomial $Z(G;x)$ is also known as the independence polynomial of $G$.

In this section we present a simple algorithm to compute the coefficients of the Taylor series of $\log(Z(G;x))$.
In what follows we will use the following basic lemma.

\begin{lemma}\label{lemma:independence}
For any graph $G$ and vertex $v$ we have
\[
    Z(G;x)=\underbrace{Z(G-v;x)}_{Z^{v \textrm{ out}}(G;x)}+\underbrace{x Z(G\setminus N[v];x)}_{Z^{v\textrm{ in}}(G;x)}.
\]
\end{lemma}

As mentioned in the introduction, instead of directly computing the coefficients of $\log(Z(G;x))$ we will actually compute the coefficients of its derivative.

It will be convenient to introduce the following notation.
For $S\subseteq V$ and $v\in S$ let us define.
\begin{equation}\label{eq:ratio ind}
R_{S,v}(x)=\frac{Z^{v \textrm{ in}}(G[S];x)}{Z^{v \textrm{ out}}(G[S];x)}
\end{equation}
and let $R_{S,v}^{[k]}(x)$ be the $k$th Taylor polynomial of $R_{S,v}(x)$ around $0$.
Using this notation, we see that by~\eqref{eq:derivative log=sum ratios},

\begin{equation}\label{eq:ratio_sum}
\left(x\frac{d}{dx}\log(Z(G;x))\right)^{[m]}
    =\sum_{v\in V} R_{V,v}^{[m]}(x)\st{m}\left(\frac{1}{1+R_{V,v}(x)}\right)^{[m]},
\end{equation}
Since $R_{S,v}(0)=0$ we can apply Lemma~\ref{lemma:reciprocial} to obtain that the right-hand side of~\eqref{eq:ratio_sum} is computable in at most $O(n m^3)$ time if we have access to 
$\{R_{V,v}^{[m]}(x)\}_{v\in V}$. 


Our main task is therefore to devise an algorithm to compute the $k$th Taylor polynomial of $R_{V,v}(x)$. This is the content of the next proposition.

\begin{proposition}\label{prop:alg ind}
Under Assumption~\ref{as:computation}, there is an $O(\Delta^{k+1} k^3)$-time algorithm that on input of a graph $G=(V,E)$ of maximum degree at most $\Delta$ a vertex $v\in V$ and $k\in \mathbb{N}$ computes $R_{G,v}^{[k]}(x)$.
\end{proposition}
\begin{proof}
We first describe an algorithm to compute $R_{S,v}^{[k]}(x)$ for a subset $S$ of $V$ and $v\in S$. 
Here the algorithm has access to $V\setminus S$ so that it can determine $N_G(v)\cap S$ by simply removing any element from $N_G(v)$ that is contained in $V\setminus S$.

If $k=0$, then by definition we can output $R_{S,v}^{[k]}(x)=0$. 
Otherwise, let $\{u_1,\dots,u_\ell\}=S\cap N_{G}(v)$. 
By Lemma~\ref{lemma:independence} we know that
\begin{align*}
R_{S,v}(x)&=\frac{Z^{\textrm{ in}}(G[S];x)}{Z^{v \textrm{ out}}(G[S];x)}=\frac{x Z(G[S\setminus N[v]];x)}{Z(G[S-v];x)}\\
&=\frac{x}{\prod_{i=1}^\ell\frac{Z(G[S\setminus\{v,u_1,\dots,u_{i-1}\}];x)}{Z(G[S\setminus\{v,u_1,\dots,u_{i}\}];x)}}=\frac{x}{\prod_{i=1}^\ell\left(1+R_{S\setminus\{v,u_1,\dots,u_{i-1}\},u_i}(x)\right)}.
\end{align*}
Therefore
\[
R_{S,v}^{[k]}(x)=x\left(\frac{1}{1+F(x)}\right)^{[k-1]},
\]
where $F(x)=\left(\prod_{i=1}^\ell1+R^{[k-1]}_{S\setminus\{v,u_1,\dots,u_{i-1}\},u_i}(x)\right)^{[k-1]}-1$.

Now let us estimate the computation time based on this formula. 
Let $\tau(k)$ denote the computational time needed to compute $R_{S,v}^{[k]}(x)$ for any $S\subseteq V$ such that $V\setminus S$ has size at most $O(k\Delta)$. 
We would like to show inductively that $\tau(k) \leq C\Delta^{k+1}k^3$ for some positive constant $C$.

Note that at the start of the algorithm $S=V$ and in every recursive call the set is decreased by at most $\Delta+1$ so that throughout $|V\setminus S|\leq k(\Delta+1)=O(\Delta k)$.
To compute $F(x)$ we need to determine $S\cap N_G(v)$, which takes time $O(\Delta^2 k)$, we need to determine $R^{[k-1]}_{S\setminus\{v,u_1,\dots,u_{i-1}\},u_i}(x)$ for $i=1\ldots, \ell$, which takes time $\ell \tau(k-1)$ and finally we need to carry out $\ell\leq \Delta$ multiplications, which takes time at most $\Delta O(k^2).$
So in total this gives a bound of $\Delta(O(k^2+O(\Delta k))+\tau(k-1))$ on the time to compute $F(x).$
By Lemma~\ref{lemma:reciprocial} we can compute $R_{S,v}^{[k]}(x)$ from $F(x)$ in time $O(k^3)$. 
This means that
\[
    \tau(k)\le O(k^3)+O(\Delta^2k)+\Delta(O(k^2)) +\Delta \tau(k-1),
\]
and thus by induction
\[
\tau(k)\le O(k(k^2+k\Delta+\Delta^2))+\Delta C(\Delta^k (k-1)^3)\leq C(\Delta^{k+1}k^3),
\]
provided $C$ is large enough.

\end{proof}

We have recorded the algorithm used in the proof of Proposition~\ref{prop:alg ind} as Algorithm 1 in pseudo code.

Theorem~\ref{thm:hard-core} now follows as a direct corollary to this proposition.
The algorithm is recorded in pseudo code as Algorithm 2 below.


\begin{algorithm}[H]
\caption{$R(G,S,v,k)$}\label{alg:indep_ratio}
\begin{algorithmic}[1]
\Require graph $G$, $S\subseteq V(G)$ $v\in S$, $k\in\mathbb{N}$
\Ensure $p(x)=R_{S,v}^{[k]}(x)$
\If{$k=0$}
    \State $p=0$
\ElsIf{$|S|=1$}
    \State $p=x$
\ElsIf{$k>0$ and $|S|>1$}
    \State Let $\{u_1,\dots u_\ell\}=N_G(v)\cap S$
    \State $S_0 \gets S\setminus \{v\}$
    
    \For{ $i=1,\dots,\ell$}
        \State $R_i \gets R(G,S_0,u_i,k-1)$ \Comment{This is the recursive call}
        \State $S_0 \gets S_0\setminus\{u_i\}$
    \EndFor
    
    \State $F(x)=1$
    \For{$i=1,\dots,\ell$}
        \State $F(x) \gets F(x)\st{k-1} (1+R_i)$
    \EndFor
    \State $F(x)\gets F(x)-1$
    \State $p(x) \gets x\left(\frac{1}{1+F(x)}\right)^{[k-1]}$ \Comment{Apply Lemma~\ref{lemma:reciprocial}}
\EndIf
\end{algorithmic}
\end{algorithm}

    
    

\begin{algorithm}[H]
\caption{$LOG\_Z(G,m)$}\label{alg:indep_log}
\begin{algorithmic}[1]
\Require graph $G$,  $m\in\mathbb{N}$
\Ensure $p=\log(Z(G;x))^{[m]}$
\If{$|V(G)|=0$}
    \State $p(x)=0$
\ElsIf{$m=0$}
    \State $p(x)=0$
\ElsIf{$m>0$ and $|S|\ge 1$}
    \State $p(x)\gets 0$
    \For{$v\in V(G)$} 
        \State $R_v(x)\gets R(G,V(G),v,m)$ \Comment{Recursive call}

        \State $F(x)=\left(\frac{1}{1+R_v(x)}\right)^{[m-1]}$  \Comment{Apply Lemma~\ref{lemma:reciprocial}}
        \State $p(x)\gets p(x)+R_v(x)\st{m-1} F(x)$
    \EndFor
    \State $p(x)\gets \sum_{k=1}^{m} \frac{[x^k]p(x)}{k}x^k$
\EndIf
\end{algorithmic}
\end{algorithm}

\section{Sink-free orientations}\label{sec:sink free}
In this section we will prove Theorem~\ref{thm:sink free}. 

We start with some additional notation 
For any graph $G$ and $U\subseteq V(G)$ let us define the polynomial
\[
Z_{\textrm{sfo},U}(G;t)=\sum_{S\in\mathcal{I}(G[U])}\prod_{v\in S}(-t^{\deg_G(v)}).
\]
It is important to stress here that even though $U$ may not be equal to $V$ we do take the weight of vertex $v$ to be equal to $-t^{\deg_G(v)}$.
As in the previous section it will be useful to consider ratios. To this end we define for a vertex $v\in U$ the following rational function,
\begin{equation}\label{eq:ratio sf}
    R_{U,v}(t):=\frac{Z^{v \text{ in}}_{\textrm{sfo},U}(G;t)}{Z^{v\text{ out}}_{\textrm{sfo},U}(G;t)}=\frac{-t^{\deg_G(v)}Z_{\textrm{sfo},U\setminus N_G[v]}(G;t)}{Z_{\textrm{sfo},U-v}(G;t)},
\end{equation}
where the superscripts $v\text{ in}$ (resp. $v\text{ out}$) indicate that we only sum over independent sets containing $v$ (resp. not containing $v$) and where we use Lemma~\ref{lemma:independence} for the second equality.

Denote by $u_1,\ldots,u_\ell$ the neighbors of $v$ in $U$. 
Similar as in the previous section we have the identity
\begin{equation}\label{eq:ratio recursion sf}
        R_{U,v}(t)=\frac{\left(-t^{\deg_G(v)}\right)Z_{\textrm{sfo},U\setminus N_G[v]}(G,t)}{Z_{\textrm{sfo},U-v}(G,t)}=\frac{-t^{\deg_G(v)}}{\prod_{i=1}^\ell (1+R_{U\setminus\{v,u_1,\dots,u_{i-1}\},u_i} (G,t))}.
\end{equation}

\subsection{A zero-free region}
The next lemma establishes a zero-free disk for $Z_{\textrm{sfo},U}(G,t)$ allowing us to use Barvinok's interpolation method.
In what follows for $r>0$, $B_r(0)$ denotes the disk of radius $r$ centered at $0$.
Let for a positive integer $\delta$,
\[
r_\delta:=\frac{(\delta-1)^{\tfrac{\delta-1}{\delta}}}{\delta}
\] 
and note that $r_3>1/2$.
\begin{lemma}
Let $\delta\ge 3$ be an integer and let $r=r_\delta$.
Then for any $t\in B_{r}(0)$ and any graph $G=(V,E)$ with minimum degree at least $\delta$ and $U\subseteq V$,
\[
Z_{\textrm{sfo},U}(G,t)\neq 0.
\]
\end{lemma}
\begin{proof}
We may assume that $G$ is connected.
We first consider the case $U\subsetneq V$.
We prove by induction on the size of $U$ the following claims:
    \begin{enumerate}
        \item[(i)] $Z_{\textrm{sfo},U}(G,t)\neq 0$,
        \item[(ii)] for each vertex $v\in U$ such that $v$ has a neighbor in $V\setminus U$, $|R_{U,v}(t)|\le 1/\delta$.
    \end{enumerate}
Note that both claims trivially hold in case $U=\emptyset$. 
Now assume $U\neq \emptyset$.
We note that item (ii) implies (i). Indeed, since $G$ is connected there exists a vertex $v\in U$ that has a neighbor in $V\setminus U$.
By induction $Z_{\textrm{sfo},U-v}(G;t)\neq 0$ and since 
\begin{equation}\label{eq:Z_U vs 1+R}
\frac{Z_{\textrm{sfo},U}(G;t)}{Z_{\textrm{sfo},U-v}(G;t)}=1+R_{U,v}(t)
\end{equation}
it follows that $Z_{\textrm{sfo},U}(G;t)\neq 0$. 

We now focus on proving (ii). To this end let $v\in U$ such that $v$ has a neighbor in $V\setminus U$.
Let $u_1,\ldots,u_\ell$ be the neighbors of $v$ inside $U$ and note that $\ell\leq \deg_G(v)-1$.
By~\eqref{eq:ratio recursion sf} and induction we have
\begin{align*}
        |R_{U,v}(t)|&=
        \left|\frac{-t^{\deg_G(v)}}{\prod_{i=1}^\ell (1+R_{U\setminus\{v,u_1,\dots,u_{i-1}\},u_i} (G,t))}\right|
        \\
        &\leq \frac{|t|^{\deg_G(v)}}{\prod_{i=1}^\ell (1-|R_{U\setminus\{v,u_1,\dots,u_{i-1}\},u_i} (G,t))|}
        \\ 
        &\leq (1-1/\delta)^{-\ell} r^{\deg_G(v)}\leq r\left(\frac{r}{1-1/\delta}\right)^{\deg_G(v)-1}
        \\
        &\leq r\left(\frac{r}{1-1/\delta}\right)^{\delta-1}=1/\delta,
        \end{align*}
where the final inequality is due to the fact that $r=r_\delta< 1-1/\delta$. 
Indeed it suffices to check that $r^\delta<(1-1/\delta)^\delta$ and this holds since $r^\delta=1/\delta(1-1/\delta)^{\delta-1}$ and $\delta\geq 3$.
This proves item (ii) of the induction.
We now finally have to check that $Z_{\textrm{sfo},V}(G;t)\neq 0.$
To this end we realize that for any vertex $v$, arguing as above, we have by induction,
\begin{align*}
    |R_{V,v}(t)|\leq \left(\frac{r}{1-1/\delta}\right)^{\deg_G(v)}\leq  \left(\frac{r}{1-1/\delta}\right)^\delta=\frac{1/\delta}{1-1/\delta}=\frac{1}{\delta-1}<1.
\end{align*}
Therefore we must have $Z_{\textrm{sfo},V}(G;t)\neq 0$ because of~\eqref{eq:Z_U vs 1+R} applied to $U=V$.
This finishes the proof.
\end{proof}

As an immediate corollary to our zero-freeness result and (the proof of) ~\cite{Barbook}*{Lemma 2.2.1} we have the following.
\begin{corollary}\label{cor:zero-free sf}
Let $\delta\geq 3$ be an integer and let $G=(V,E)$ be an $m$-edge graph of minimum degree at least $\delta$.
Denote by $f_k(t)$ the $k$th order Taylor approximation to $\log(Z_{\textrm{sfo},V}(G;t))$.
Then 
\[
|f_k(1/2)-\log(Z_{\textrm{sfo},V}(G;1/2))|\leq \frac{2m}{(k+1)(2r_\delta-1)(2 r_\delta)^k}.
\]
\end{corollary}
\begin{proof}
This follows directly once we realize that the degree of the polynomial $Z_{\textrm{sfo},V}(G;t)$ is bounded by $\sum_{v\in V}\deg_G(v)=2m$.
\end{proof}

For $f_k(1/2)$ to be a $\varepsilon$-approximation it suffices to choose $k\geq \frac{\log(m/\varepsilon)}{\log(2r_\delta)}$.
When $\delta=3$ this gives us that $k\geq 17.66\log(m/\varepsilon)$, if $\delta\geq 5$ it suffices to take $k\geq 5.2\log(m/\varepsilon)$,
while if $\delta$ tends to $\infty$ we have that $r_\delta\to 1$ and hence it suffices to have $k\geq 1.45\log(m/\varepsilon)$ when $\delta$ is large enough.

\subsection{The algorithm}
We will use the same framework as established in the previous section to obtain an FPTAS for $\textrm{sfo}(G)=Z_{\textrm{sfo},V(G)}(G,1/2)$. So consider
\begin{equation}
     t\frac{d }{d t}\log(Z_{\textrm{sfo},V}(t))=\sum_{v\in V}\left(-\deg_G(v)t^{\deg_G(v)}\right)\frac{Z_{\textrm{sfo},V\setminus N_G[v]}(G,t)}{Z_{\textrm{sfo},V}(G,t)}=\sum_{v\in V} \frac{\deg_G(v)R_{V,v}(t)}{1+R_{V,v}(t)}.\label{eq:der log sink free}  
\end{equation}

As in the previous section, our main goal will be to compute the Taylor coefficients of the ratios $R_{V,v}(t).$

\begin{proposition}\label{prop:alg sf}
Under Assumption~\ref{as:computation}, there is an $O(e^{k/e})$-time algorithm with input a graph $G=(V,E)$, a vertex $v\in V$ and $k\in \mathbb{N}$ computes $R_{V,v}^{[k]}(t)$.
\end{proposition}
\begin{proof}
We first describe an algorithm to compute $R_{U,v}^{[k]}(x)$ for a subset $U$ of $V$ and $v\in U$. 
Here the algorithm has access to $V\setminus U$ so that it can determine $N_G(v)\cap U$ by simply removing any element from $N_G(v)$ that is contained in $V\setminus U$.

If $k=0$, then by definition we output $R_{U,v}^{[k]}(t)=0$. 
Otherwise, let $\{u_1,\dots,u_\ell\}=U\cap N_G(v)$. Then by~\eqref{eq:ratio recursion sf},
    \[
        R_{U,v}(t)=\frac{\left(-t^{\deg_G(v)}\right)Z_{U\setminus N_G[v]}(G,t)}{Z_{U\setminus \{v\}}(G,t)}=\frac{-t^{\deg_G(v)}}{\prod_{i=1}^\ell (1+R_{U\setminus\{v,u_1,\dots,u_{i-1}\},u_i} (G,t))},
    \]
    therefore
    \[
    R_{U,v}^{[k]}(t)=-t^{\deg_G(v)}\left(\frac{1}{1+F(t)}\right)^{[k-\deg_G(v)]},
    \]
    where $F(t)=\left(\prod_{i=1}^\ell (1+R_{U\setminus\{v,u_1,\dots,u_{i-1}\},u_i} (G,t))^{[k-\deg_G(v)]}\right)^{[k-\deg_G(v)]}-1$.

So in case $\deg_G(v)> k$, we  output $0$. 
Note that it only takes $O(k)$ time to decide whether $\deg_G(v)> k$ and to list the set $\{u_1,\ldots,u_\ell\}$ in case $\deg_G(v)\leq k$.
In that case, to compute $R^{[k]}_{U,v}(t)$ we need to compute $F(t)$ and use Lemma~\ref{lemma:reciprocial} to compute $R^{[k]}_{U,v}$ (which takes time $O((k-\deg_G(v))^3)$.)
To compute $F(t)$ we need to call the algorithm to compute $R^{[k-\deg_G(v)]}_{U\setminus\{v,u_1,\dots,u_{i-1}\},u_i}$ for $i=1,\ldots,\ell$ and multiply the resulting polynomials.

Let us now denote by $\tau(k)$ the time to compute $R^{[k]}_{U,v}(t)$. We will inductively show that $\tau(k)$ is bounded by $Ce^{k/e}$ for a large enough constant $C$.
First of all, we may assume that $\deg_G(v)\leq k$, otherwise the algorithm only takes linear time.
Then by the reasoning above we have
\begin{align}
        \tau(k)&\le O(k)+\deg_G(v)(O((k-\deg_G(v))^2)+ \tau(k-\deg_G(v))) + O((k-\deg_G(v))^3)\nonumber
        \\
        &\le O(k)+\max_{1\le d\le k} d O((k-d)^2)+ d C e^{(k-d)/e} + O((k-d)^3)\label{eq:maximization}
        \\
        &\le C\left(\tfrac{3}{e^{3/e}}e^{k/e}\right)+O(k^3)\le Ce^{k/e}\nonumber,
    \end{align}
    which is true if $C$ is chosen sufficiently large enough.
\end{proof}
\begin{remark}\label{rem:choice of constant}
Note that for graphs of (large) minimum degree $\delta\geq 3$ the exponent of $e$ can be replaced by something smaller than $k/e$, thereby reducing the running time. In particular, we can replace $k/e$ by $k(1+\eta)\log(\delta)/\delta$ for any $\eta>0$.
\end{remark}

We can now prove Theorem~\ref{thm:sink free}.
\begin{proof}[Proof of Theorem~\ref{thm:sink free}]
By Corollary~\ref{cor:zero-free sf} and Lemma~\ref{lem:sink free}, to compute a relative $\exp(\varepsilon)$ approximation to $\text{sfo}(G)$ it suffices to compute the Taylor polynomial of $\log(Z_{\textrm{sfo},V}(G;t))$ up to order $k=(17.66)\log(m/\varepsilon)$.
By the algorithm of Proposition~\ref{prop:alg sf} we can compute the ratios  $(R_{V,v})^{[k]}$ for $v\in V$ in time $O(ne^{k/e})=O(n(m/\varepsilon)^{6.5})$.
We then use~\eqref{eq:der log sink free} in combination with Lemma~\ref{lemma:reciprocial} to compute the Taylor polynomial of $\log(Z_{\textrm{sfo},V}(G;t))$ up to order $k$ in time $O(k^3)$. 
The overall running time is bounded by $O(n(m/\varepsilon)^{7})$, as desired.
\end{proof}
\begin{remark}
Following up on the previous remark.
The running time of the algorithm in Theorem~\ref{thm:sink free} can be improved to 
\[O\left(n(m/\varepsilon)^{(1+\eta)\tfrac{\log(\delta)}{\delta\log(2r_\delta)}}\right)\] 
for any positive $\eta$, for graphs of minimum degree at least $\delta$.
So for large $\delta$ we obtain a near linear (in terms of the number of vertices) time algorithm.
\end{remark}

\section{Chromatic Polynomial}\label{sec:chromatic}
In this section we will describe our algorithm for computing the low order coefficients of $P(G;z)=(-z)^{|V(G)|}\chi(G;-1/z)$ for a graph $G=(V,E)$.
We start with some definitions needed to provide a combinatorial interpretation of the coefficients of $P(G;z)$.

Fix a total order $<$ of the edges of $G$.
Recall that a set $F\subseteq E$ is called \emph{broken circuit free} (abbreviated BCF) if $(V,F)$ is a forest (i.e. contains no cycles) and for each edge $e\notin F$ such that $F\cup e$ contains a cycle $C$, the edge $e$ is not the largest edge of $C$.
Let $\mathcal{F}^{<}_G$ be the collection of BCF subgraphs of $G$. Then we know by Whitney's theorem~\cite{Whitney} that
\begin{equation}\label{eq:Whitney}
    P(G;z)=\sum_{F\in\mathcal{F}^{<}_G}z^{|F|}.
\end{equation}
Note that while $\mathcal{F}_G^{<}$ actually depends on the ordering $<$, but by definition, the polynomial $P(G;z)$ does not.

Our aim is to compute the low-order coefficients of the Taylor series of $\log(P(G;z)$ at $z=0$.
We know that
\begin{align}
    z\frac{d}{d z}P(G;z)&=\sum_{e\in E(G)}\underbrace{\sum_{e\in F\in\mathcal{F}^{<}_G} z^{|F|}}_{P^{e\textrm{ in}}(G;z)},
\end{align}
and therefore
\begin{equation}\label{eq:der to rational chrom}
    \left(z\frac{d}{d z} \log(P(G;z))\right)^{[m]}=\sum_{e\in E(G)}\left(\frac{P_e^\textrm{ in}(G;z)}{P(G;z)}\right)^{[m]}.
\end{equation}
As before we aim to show a procedure to compute for each edge $e\in E(G)$ the Taylor polynomial of
\[
    \frac{P^{e \textrm{in}}(G;z)}{P(G;z)}=\sum_{T: e\in T\in\mathcal{T}^<_G} z^{|T|}\frac{\sum_{F\in\mathcal{F}^<_{G\setminus V(T)}}z^{|F|}}{P(G;z)}=\sum_{T: e\in T\in\mathcal{T}^<_G} z^{|T|}\frac{P(G\setminus V(T);z)}{P(G;z)},
\]
here $\mathcal{T}_G^{<}$ denotes the collection of subtrees of $G$ that are BCF.
Let us define for any $S\subseteq V(G)$ and $v\in S$ the rational function 
\begin{equation}\label{eq:ratio chrom}
 R_{S,v}(z)=\frac{P(G[S-v];z)}{P(G[S];z)}.    
\end{equation}
Writing for $T\in\mathcal{T}^<_G$ such that $e\in T$, $V(T)=\{v_1,\ldots,v_\ell\}$. Then we can write $\frac{P(G\setminus V(T);z)}{P(G;z)}=\prod_{i=1}^{\ell}R_{V(G)\setminus \{v_1,\ldots,v_{i-1}\},v_i}$.
Therefore
\begin{equation} \label{eq:e is in}
 \left(\frac{P^{e \textrm{ in}}(G;z)}{P(G;z)}\right)^{[m]}=\sum_{\substack{T\in\mathcal{T}^<_G\\e\in T,|T|\leq m}}z^{|T|}\prod_{i=1}^{\ell}R^{[m-|T|]}_{V(G)\setminus \{v_1,\ldots,v_{i-1}\},v_i}(z).
 \end{equation}
To compute the sum above, we will need to be able to efficiently enumerate all BCF trees containing a specific edge or vertex of size at most $m$ and a procedure to compute the functions $R^{[m]}_{S,v}(z)$. 
The first task is handled by the next lemma, whose proof we postpone to the next subsection.

 \begin{lemma}\label{lem:enmurate BCF tree}
 Under Assumption~\ref{as:computation}, there exists an algorithm that on input of a graph $G=(V,E)$ of maximum degree $\Delta$ with a given total order of its edges, a vertex $v\in V$ (resp. an edge $f\in E$) and $m\in \mathbb{N}$ that enumerates all BCF trees $T\subset E$ that contain $v$ (resp. $f$) in time $O((e\Delta)^m\Delta^2 m^3)$.
\end{lemma}

The next lemma handles the computation of the coefficient of the ratios.

\begin{lemma}\label{lem:compute R chrom}
Under Assumption~\ref{as:computation} there exists an algorithm that on input of a graph $G=(V,E)$ of maximum degree at most $\Delta$, a vertex $v\in V$ and $m\in \mathbb{N}$ computes
\[
 R_{V,v}^{[m]}(z)
\]
in time $O((3e\Delta)^mm^3)$.  
\end{lemma}
\begin{proof}
We first describe an algorithm for computing $R_{S,v}(z)$ for $S\subseteq V$.
For this we will need the following recursion for $R_{S,v}(z)$.
By definition we have,
    \begin{align*}
        1/R_{S,v}(z)&=\frac{P(G[S];z)}{P(G[S-v];z)}=\sum_{F\in \mathcal{F}^{<}_{G[S]}}\frac{z^{|F|}}{P(G[S-v];z)}\\
        &=\sum_{\substack{F\in \mathcal{F}^{<}_{G[S]}\\ v\notin V(F)}}\frac{z^{|F|}}{P(G[S-v];z)}+\sum_{\substack{F\in \mathcal{F}^{<}_{G[S]}\\ v\in V(F)}}\frac{z^{|F|}}{P(G[S-v];z)}\\
        &=1+\sum_{\substack{F\in \mathcal{F}^{<}_{G[S]}\\ v\in V(F)}}\frac{z^{|F|}}{P(G[S-v];z)}\\
        &=1+\sum_{\substack{T:T\in\mathcal{T}^{<}_{G[S]}\\ v\in V(T),|T|\ge 1}}z^{|T|}\frac{P(G[S-V(T)];z)}{P(G[S-v,z])}.
    \end{align*}
    This means that
    \begin{equation}\label{eq:from F to R chrom}
        R_{S,v}^{[m]}(z)=\left(\frac{1}{1+F^{[m]}(z)}\right)^{[m]},
    \end{equation}
    where
    \begin{equation}\label{eq:def F chrom}
        F^{[m]}(z)=\sum_{\substack{T:T\in\mathcal{T}^{<}_{G[S]}\\ v\in V(T),1\le |T|\le m}}z^{|T|}\left(\frac{P(G[S-V(T)];z)}{P(G[S-v];z)}\right)^{[m-|T|]}.
    \end{equation}
Denote for a tree $T$ appearing in~\eqref{eq:def F chrom} $\{v,v_1,\dots,v_\ell\}=V(T)$. Then we can express
    \begin{equation}\label{eq:prod chrom}
        \left(\frac{P(G[S-\{v_1,\dots,v_\ell\}];z)}{P(G[S-v];z)}\right)^{[m-|T|]}=\left(\prod_{i=1}^\ell \left(R_{S\setminus\{v,v_1,\dots,v_{i-1}\},v_i }(z)\right)^{[m-|T|]}\right)^{[m-|T|]}.
    \end{equation}

Now let $\tau(m)$ denote the time needed to calculate $R_{S,v}^{[m]}$ for any $S\subseteq V(G)$ such that $|V\setminus S|\leq m$ and $v\in S$.
Our aim is to inductively show that $\tau(m)=C(3e\Delta)^m m^3$ for some constant $C>0$.

To compute  $R_{S,v}^{[m]}$ from $F^{[m]}(z)$ takes time $O(m^3)$ by Lemma~\ref{lemma:reciprocial}.
To compute $F^{[m]}(z)$ we need to enumerate all BCF-trees $T$ rooted at $v$ contained in $G[S]$ with at most $m$ edges and compute $\left(\frac{P(G[S-V(T)];z)}{P(G[S-v];z)}\right)^{[m-|T|]}$.
The latter can be computed as a product of $|T|$ ratios with $|T|$ calls to the algorithm, giving a running time of $O(|T|(m-|T|)^2)\tau(m-|T|){|T|}$.
The enumeration of the BCF-trees can be done in time $O((e\Delta)^m\Delta^2m^3)$ by the previous lemma.
By~\cite{Borgsetal}*{Lemma 2.1} the number of trees rooted at a $v$ with $k$ edges is at most $(e\Delta)^{k}$.
Putting this all together, we obtain the following bound

    \begin{align*}
    \tau(m)&\le O(m^3) + (e\Delta)^m O(\Delta^2 m^3)+\sum_{k=1}^{m} (e\Delta)^k \left( \tau(m-k) k +k O((m-k)^2)\right)\\
        &\le O(m^3)+(e\Delta)^m O(\Delta^2 m^3)+C(3e\Delta)^m\sum_{k=1}^{m}3^{-k}(m-k)^3k+\sum_{k=1}^{m}k(e\Delta)^k O((m-k)^2)  \\
        &\le O(m^3)+(e\Delta)^m O(\Delta^2 m^3)+\tfrac{3}{4}C(3e\Delta)^mm^3\\
        &\le C(3e\Delta)^mm^3
    \end{align*}
 where we use induction in the second inequality; the third inequality uses that $\sum_{k\geq 1}kx^k=x\tfrac{d}{dx}\tfrac{1}{1-x}=\tfrac{x}{(1-x)^2}$ for $x=1/3$ and $x=\tfrac{1}{e\Delta}$ and  the last inequality is true provided $C$ is sufficiently large.
 This proves the lemma.


\end{proof}

We can now provide a proof of Theorem~\ref{thm:chrom}
\begin{proof}[Proof of Theorem~\ref{thm:chrom}]
The first $m$ coefficients of $\log(P(G;z))$ can easily be computed from the coefficients of $\frac{d}{dz}\log(P(G;z))$ in linear time.
By~\eqref{eq:der to rational chrom} we need to compute the first $m$ coefficients of $\frac{P^{e \textrm{ in}}(G;z)}{P(G;z)}$ for each edge $e$ of $G$.
By~\eqref{eq:e is in} this computation boils down to enumerating all BCF trees $T$ of size at most $m$ that contain $e$ which can be done in time $O((e\Delta)^m\Delta^2m^3)$ by Lemma~\ref{lem:enmurate BCF tree}. 
For each BCF tree $T$ with vertex set $\{v_1,\ldots,v_\ell\}$ we need to compute the terms in the product
\[
\prod_{i=1}^\ell \left(R_{S\setminus\{v_1,\dots,v_{i-1}\},v_i }(z)\right)^{[m-|T|]}
\]
and the actual product, which can be done in time $O((3e\Delta)^{m-|T|}{(m-|T|)}^3)$ by Lemma~\ref{lem:compute R chrom}.
All together, this takes $O((3e\Delta)^m m^3)$ time, reasoning as in the proof of Lemma~\ref{lem:compute R chrom}.
\end{proof}




\subsection{Enumeration of trees}
In this section we will describe an algorithm that enumerates bounded size trees containing a specific vertex. 
In particular, we will provide a proof of Lemma~\ref{lem:enmurate BCF tree} at the end of this section.

Let $G=(V,E)$ be a graph of maximum degree $\Delta$ and $v$ to be a vertex of $G$. We assume that $E$ is equipped with a total order $<$.
Let $\mathcal{Q}_v$ be the graph on the collection of subtrees of $G$ containing $v$, where two trees are connected if and only if they differ by exactly $1$ edge. Observe that in this graph the $m$th neighborhood of the single vertex tree $\{v\}$ is exactly the set of trees of at most $m$ edges. 
To do an efficient enumeration of this neighborhood in $\mathcal{Q}_v$, we will construct a spanning tree $\mathcal{T}_v$ of $\mathcal{Q}_v$ on which we can perform efficiently a graph search algorithm.

First let us describe the construction of $\mathcal{T}_v$. 
For each subtree $T\in V(\mathcal{Q}_v)$ let us define a unique identifier: if $T$ has exactly $1$ edge $e$, then $ID(T)=(e)$. 
If $T$ consist of $m\ge 2$ edges, then $ID(T)=(e_1,\dots,e_m)$, where $e_m$ is the largest leaf edge of $T$, such that $v$ is not an isolated vertex of $T-e_m$ and $(e_1,\dots,e_{m-1})=ID(T-e_m)$.
By definition $ID(T-e_m)$ is exactly the sequence formed by the first $m-1$ elements of the sequence $ID(T)$; we call $T-e_m$ to be the parent of $T$. 
In particular, if $T'$ is a parent of $T$, then $(T,T')\in E(\mathcal{Q}_v)$. 
Thus, if we only take the edges, that connects every tree with their parent, then we obtain a spanning tree $\mathcal{T}_v$ of $\mathcal{Q}_v$.

The following lemma shows that we can perform graph search on the tree $\mathcal{T}_v$ efficiently.
\begin{lemma}\label{lem:compute all trees}
Under Assumption~\ref{as:computation}, there exists an algorithm that on input of a graph $G$ of maximum degree at most $\Delta$ equipped with a total order of its edges, $v\in V(G)$ and $T\in V(\mathcal{Q}_v)$,
\begin{itemize}
\item computes $ID(T)$ in time $O(|T|^2)$;
\item computes  $\partial V(T)$, the set of edges leaving $V(T)$ in the graph $G$, in time $O(\Delta |T|)$;
\item computes $\{ID(T')~|~T'\in\mathcal{Q}_v, \textrm{$T$ is the parent of $T'$}\}$ in time $O(\Delta |T|)$.
\end{itemize}
\end{lemma}
\begin{proof}
    The first part follows from the definition, since finding all leaf edges and vertices of a tree takes $O(m)$ time. Therefore the choice of  $e_m$ can be made in time at most $O(m)$. 

    For the second part, we have to perform a simple DFS algorithm in $G$ from $v$ with the constraint that we are only allowed to traverse edges of $T$. 
    Each time we identify edges not in $T$ and that are not induced by $V(T)$ we know that these belong to $\partial V(T)$.
    In this process, we will identify all edges $e\in E(G)$, such that $e\in \partial V(T)$ and this takes at most $O(m\Delta)$ time.

    Now let us prove the third part. If $T'$ is a child of $T$, then it means that the last edge $e_{m+1}$ of $ID(T')$ is in $\partial V(T)$. Moreover, either $e_{m+1}$ is larger than $e_m$ or $e_{m+1}$ is incident to $e_m$ and larger than the second largest leaf of $E(T)$.
    Therefore, examining one by one the edges in $\partial V(T)$, which has size at most $\Delta m$, we can identify all children of $T$.
\end{proof}

As a corollary of this lemma we obtain the following proposition.

\begin{proposition}\label{prop:enumerate tree}
  Under Assumption~\ref{as:computation}, there exists an algorithm that on input of a graph $G=(V,E)$ of maximum degree $\Delta$ with a given total order of its edge,  a vertex $v\in V$ (resp. an edge $f\in E$) and $m\in \mathbb{N}$ that enumerates all trees $T\subset E$ that contain $v$ (resp. $f$) in time $O((e\Delta)^m\Delta m)$ time.  
\end{proposition}

\begin{proof}
Let us first consider the case that the trees must contain $v$.
Let $\mathcal{T}_{v,\le m}$ be the subtree of $\mathcal{T}_v$, that is spanned by the trees of at most $m$ edges. 
By the previous lemma, at any vertex of $\mathcal{T}_{v,\le m}$ we can compute the possible neighbors in $O(\Delta m)$ time.
Since the size of $\mathcal{T}_{v,\le m}$ is at most $(e\Delta)^m$ by~\cite{Borgsetal}*{Lemma 2.1}, we can therefore enumerate all vertices of $\mathcal{T}_{v,\le m}$ in time $(e\Delta)^mO(\Delta m)$, as desired.

In case our trees must contain the edge $f$ the argument is similar. We just replace $\mathcal{Q}_v$ by $\mathcal{Q}_f$, the graph on the collection of subtrees of $G$ that contain $f$ and appropriately define $\mathcal{T}_f$. 
We leave the remaining details to the reader. 
\end{proof}

For a graph $G=(V,E)$ of with a given total ordering of its edges $<$ and a forest $F\subset E$ we call an edge $f\in E\setminus F$ \emph{broken} if $F\cup f$ contains a cycle $C$ and $f$ is the largest edge in that cycle. We denote the collection of broken edges of $F$ by $B(F).$

\begin{lemma}\label{lem:compute broken}
Under Assumption~\ref{as:computation}, there exists an algorithm that on input of a graph $G$ of maximum degree at most $\Delta$ equipped with a total order of its edges, $v\in V(G)$ and any $T\in \mathcal{Q}_v$, 
\begin{itemize}
\item computes $E_G(V(T))$, the set of induced edges by $V(T)$ in the graph $G$, in time $O(\Delta |T|)$;
\item computes $B(T)$, the set of broken edges of $T$, in time $O(\Delta |T|^2)$.
\end{itemize}
\end{lemma}
\begin{proof}
The first item follows by a DFS algorithm just as in the proof of Lemma~\ref{lem:compute all trees}.
The second item follows by running comparing each edge that is added in the procedure of determining all induced edges with all other edges in the cycle that is formed (which is found by finding the first common ancestor of the two endpoints of the edge). This clearly takes time at most $O(|T|)$ for each such edge.
\end{proof}

Lemma~\ref{lem:enmurate BCF tree} now follows by combining Proposition~\ref{prop:enumerate tree} with the previous lemma.

\section{Graph homomorphisms}\label{sec:graph hom}
Throughout we fix $A$ to be an $q\times q$ symmetric matrix. 
For a graph $G=(V,E)$ we write
\[
H(G;x)=q^{-|V|}\hom(G,J+x(A-J)).
\]
In this section we will describe our algorithm for computing the low-order coefficients of $\log(H(G,x))$.

We first need a few definitions to realize $H(G,x)$ as a forest generating function.
Consider a total ordering $<$ on the edges of $G$. 
Let for a tree $T\subseteq E$
\[
w_T=w_T(x)=q^{-(|T|+1)}\sum_{\phi:V(T)\to[q]}x^{|T|}\prod_{ij\in T}(A-J)_{\phi(i),\phi(j)}\prod_{ij\in E(V(T))_T}(J+x(A-J))_{\phi(i),\phi(j)},
\]
where $E(S)_T$ denote the collection of broken edges of $T$ in the graph $G[S]$.
Next define
\[
F_G(w)=F_G((w_T)_{T \text{ tree}})=\sum_{\substack{F\in \mathcal{F}_G}}\prod_{T\in \mathcal{C}(F)}w_T,
\]
where $\mathcal{F}_G$ denotes the collection of forests of $G$ and $\mathcal{C}(F)$ denotes the collection of non-trivial components of $F$ (i.e. components with at least one edge).

\begin{lemma}\label{lem:forest graph hom} For any graph $G$ with a total ordering on its edges we have
\[
F_G(w_T)=H(G;x),
\]
where $F_G(w_T)$ is defined above.
\end{lemma}
We postpone the proof of the lemma to the end of this section, noting that it can be proved along very similar lines as~\cite{bencs2024approximatingvolumetruncatedrelaxation}*{Proposition 2.3}.

Now let us express the Taylor series of $x\frac{d}{d x}\log(H(G;x))$. We have that
\begin{align*}
x\frac{d}{d x}H(G;x)&=\sum_{T\in\mathcal{T}_G}\underbrace{(x \frac{d}{d x} w_T(x)) \cdot  H(G-V(T);x)}_{H_T(G;x)}
\\
&=\sum_{e\in E}\underbrace{\sum_{\substack{T\in\mathcal{T}_G\\ e=\min(T)}}(x \frac{d}{d x} w_T(x)) \cdot  H(G-V(T);x)}_{H^{e}(G;x)}.
\end{align*}
 Therefore,
\[
    \left(x\frac{d}{d x}\log(H(G;x))\right)^{[m]}=\sum_{\substack{T\in\mathcal{T}_G\\ 1\le|T|\le m}}\left(\frac{H_T(G;x)}{H(G;x)}\right)^{[m]}=\sum_{e\in E }\sum_{\substack{T\in\mathcal{T}_G\\ |T|\le m\\ e=\min(T)}}\left(\frac{H_T(G;x)}{H(G;x)}\right)^{[m]}.
\]
As before we aim to describe a procedure to compute 
\[
    \left(\frac{H_T(G;x)}{H(G;x)}\right)^{[m]}=(x\tfrac{d}{d x}w_T(x))^{[m]}\left(\frac{H(G-V(T);x)}{H(G;x)}\right)^{[m-|T|]}.
\]
Note that in this equation we used the fact that $x\tfrac{d}{d x}\left(w_T(x)\right)=O(x^{|T|})$.
For a tree $T$ on $k$ vertices the contribution $w_T(x)^{[m]}$ can be computed in time $O(\Delta k^2)+O(q^{k}(\Delta k)(m-k))=O(q^k \Delta k m)$, since if $k>m$ we have $w_T(x)^{[m]}=0$ otherwise after calculating the broken edges we take all possible $V(T)\to [q]$ colorings and compute the corresponding contributions.  

Let us define for any $S\subseteq V(G)$ and $v\in S$ the rational function 
\begin{equation}\label{eq:ratio graph hom}
 R_{S,v}(x)=\frac{H(G[S-v],x)}{H(G[S],x)}.   
\end{equation} 
Then for any $\{v_1,\ldots,v_\ell\}\subseteq S$ we have
\[
\frac{H(G[S-\{v_1,\dots,v_\ell\}];x)}{H(G[S];x)}=\prod_{i=1}^\ell R_{S\setminus\{v_1,\dots,v_{i-1}\},v_i }(x).
\]
This means that for any tree $T$ if we let $V(T)=\{v_1,\dots,v_\ell\}$, then
\[\left(\frac{H^e(G;x)}{H(G;x)}\right)^{[m]}=\sum_{\substack{T\in\mathcal{T}_G\\|T|\le m\\e=\min(T)}}\left(x\tfrac{d}{d x}w_T(x)\right)^{[m]}\prod_{i=1}^\ell R^{[m-|T|]}_{V(G)\setminus \{v_1,\dots,v_{i-1}\},v_i}(x).\]
By Proposition~\ref{prop:enumerate tree} we can enumerate the trees involved in the summation in time at most $O((e\Delta)^m\Delta m)$. Using the next lemma we obtain the desired algorithm.

\begin{lemma}\label{lem:alg graph hom}
  Under Assumption~\ref{as:computation} there exists an algorithm that on input of a graph $G=(V,E)$ of maximum degree at most $\Delta$ with a given total ordering of its edges, a vertex $v\in V$ and $m\in \mathbb{N}$ computes
\[
 R_{V,v}^{[m]}(x)
\]
in time $O(((q+1)e\Delta)^mm^3)$.
\end{lemma}
\begin{proof}
    The proof essentially follows the same line of proof as our proof of Lemma~\ref{lem:compute R chrom}. 

We first describe an algorithm for computing $R_{S,v}(x)$ for $S\subseteq V$.
For this we will need the following recursion for $R_{S,v}(x)$.
For an arbitrary ordering on the induced edges of $S$ we have
    \begin{align*}
        \frac{1}{R_{S,v}(x)}&=\frac{H(G[S];x)}{H(G[S-v];x)}=\sum_{F\in \mathcal{F}_{G[S]}}\prod_{T\in \mathcal{C}(F)}\frac{w_T}{H(G[S-v];x)}\\
        &=\sum_{\substack{F\in \mathcal{F}_{G[S]}\\ v\notin V(F)}}\prod_{T\in \mathcal{C}(F)}\frac{w_T}{H(G[S-v];x)}+\sum_{\substack{F\in \mathcal{F}_{G[S]}\\ v\in V(F)}}\prod_{T\in \mathcal{C}(F)}\frac{w_T}{H(G[S-v];x)}\\
        &=1+\sum_{\substack{F\in \mathcal{F}_{G[S]}\\ v\in V(F)}}\prod_{T\in \mathcal{C}(F)}\frac{w_T}{H(G[S-v];x)}\\
        &=1+\sum_{\substack{T:T\in\mathcal{T}_{G[S]}\\ v\in V(T),|T|\ge 1}}w_T\frac{H(G[S-V(T)],x)}{H(G[S-v];x)}
    \end{align*}
    This means that
    \[
        R_{S,v}^{[m]}(x)=\frac{1}{1+F^{[m]}(x)},
    \]
    where
    \[
        F^{[m]}(x)=\sum_{\substack{T:T\in\mathcal{T}_{G[S]}\\ v\in V(T),1\le |T|\le m}}w_T^{[m]}\left(\frac{H(G[S-V(T)];x)}{H(G[S-v];x)}\right)^{[m-|T|]}.
    \]
If we denote $V(T)=\{v,v_1,\dots,v_\ell\}$ for a $T$ appearing in this sum, then we can express
    \begin{equation}\label{eq:prod}
        \left(\frac{H(G[S-\{v,v_1,\dots,v_\ell\}];z)}{H(G[S];z)}\right)^{[m-|T|]}=\prod_{i=1}^\ell \left(R_{S\setminus\{v,v_1,\dots,v_{i-1}\},v_i }(z)\right)^{[m-|T|]}.
    \end{equation}

Now let us denote by $\tau(m)$ the time needed to calculate $R_{S,v}^{[m]}(z)$ for any $S\subseteq V(G)$ such that $|V\setminus S|\le m$ and $v\in S$.
It is our aim to show by induction that $\tau(m)\leq C(e(q+1)\Delta)^mm^3$.
 
   To compute $R_{S,v}^{[m]}$ from $F^{[m]}$ takes time $O(m^3)$ by Lemma~\ref{lemma:reciprocial}. 
    To compute $F^{[m]}(z)$ we need to enumerate all trees $T$ in $G[S]$ containing $v$ with at most $m$ edges and compute for each such $T$ the polynomials $w_T^{[m]}$ and $\left(\frac{H(G[S-\{v,v_1,\dots,v_\ell\}];z)}{H(G[S];z)}\right)^{[m-|T|]}$ and their product.
    The enumeration can be done in $(e\Delta)^mO(\Delta m)$ time by Proposition~\ref{prop:enumerate tree}.
   As remarked earlier, the computation of $w_T^{[m]}$ can be done in $O(q^k\Delta km)$ time for a tree with $k$ vertices. 
   By~\eqref{eq:prod}, $\left(\frac{H(G[S-\{v,v_1,\dots,v_\ell\}];z)}{H(G[S];z)}\right)^{[m-|T|]}$ can be computed as product of $\ell=|T|$ polynomials with $\ell$ calls to the algorithm, giving a running time $O(\ell(m-\ell)^2)+\tau(m-\ell)\ell$. As the number of trees containing a given vertex of $\ell$ edges is at most $(e\Delta)^\ell$ by~\cite{Borgsetal}*{Lemma 2.1}, we can bound $\tau(m)$ by

    \begin{align*}
    &O(m^3) + (e\Delta)^mO(\Delta m)+\sum_{k=1}^{m} (e\Delta)^k \left(O(q^{k+1}\Delta k(m-k))+ \tau(m-k) k +k O((m-k)^2)\right)\\
        \le &O(m^3)+(e\Delta)^mO(\Delta m)+((q+1)e\Delta)^m\left((qe\Delta)m+\tfrac{3}{4}Cm^3+m^2\tfrac{e\Delta}{(q+1)^m}\right)\\
        \le &O(m^3)+(e\Delta)^m O(\Delta m)+\tfrac{4}{5}C((q+1)e\Delta)^mm^3\\
        \le &C((q+1)e\Delta)^mm^3
    \end{align*}
    where the last two inequalities follow provided $C$ is sufficiently large and using the fact that $ \sum_{k\geq 1}kx^k=\tfrac{x}{(1-x)^2}$ and where for we use induction in the second inequality.
This proves the lemma.
\end{proof}

In exactly the same as in the proof of Theorem~\ref{thm:chrom} we can now deduces Theorem~\ref{thm: graph hom}. We leave the details to the reader.

\subsection{Proof of Lemma~\ref{lem:forest graph hom}}
Here we provide a proof of Lemma~\ref{lem:forest graph hom}, based on~\cite{bencs2024approximatingvolumetruncatedrelaxation}*{Section 2}.

We require the following lemma of Penrose~\cite{penrose}.
\begin{lemma}
Let $G=(V,E)$ be a graph with a fixed total ordering of the edges. Then the collection of connected subgraphs of $G$ can be partitioned into intervals of the form $\{F\mid T\subseteq F\subseteq T\cup B(T)\}$, where T is a tree.
\end{lemma}
As remarked in~\cite{bencs2024approximatingvolumetruncatedrelaxation} this lemma can be proved by assigning a connected subgraph to its minimum weight spanning tree.

We can now prove Lemma~\ref{lem:forest graph hom}.
\begin{proof}[Proof of Lemma~\ref{lem:forest graph hom}]
By expanding the product we can write 
\begin{align*}
H(G;x)&=q^{-|V|}\hom(G,J+x(A-J))=q^{-|V|}\sum_{\phi:V\to [q]}\prod_{uv\in E}(J+x(A-J))_{\phi(u),\phi(v)}
\\
&=\sum_{\phi:V\to [q]}q^{-|V|}\sum_{F\subseteq E}\prod_{uv\in F}(x(A-J))_{\phi(u),\phi(v)}
\\
&=\sum_{F\subseteq E}q^{-|V(F)|}\sum_{\phi:V(F)\to [q]}\prod_{uv\in F}(x(A-J))_{\phi(u),\phi(v)},
\end{align*}
where $V(F)$ denotes the collection of vertices incident with some edge of $F$.

Clearly, for a fixed set $F\subseteq E$ its contribution to the sum is multiplicative over its components. 
By the previous lemma we can uniquely realize each connected subgraph $C$ as $C=T\cup B$ with $B\subseteq B(T)$, the broken edges of $T$.
This implies that we can write 
\begin{align*}
H(G;x)&=\sum_{F\in \mathcal{F}_G}\prod_{T\in \mathcal{C}(F)}q^{-(|T|+1})\sum_{\phi:V(T)\to [q]}\sum_{B\subseteq E(T)}\prod_{uv\in T\cup B}(x(A-J))_{\phi(u),\phi(v)}
\\
&=\sum_{F\in \mathcal{F}_G}\prod_{T\in \mathcal{C}(F)}\frac{x^{|T|}}{q^{|T|+1}}\sum_{\phi:V(T)\to [q]}\prod_{uv\in T}(A-J)_{\phi(u),\phi(v)}\cdot\prod_{uv\in B(T)}(J+x(A-J))_{\phi(u),\phi(v)}
\\
&=\sum_{F\in \mathcal{F}_G}\prod_{T\in \mathcal{C}(F)}w_T,
\end{align*}
as desired.
\end{proof}

\bibliographystyle{numeric}
\bibliography{algorithmic}

\begin{bibdiv}
\begin{biblist}

\bib{anand2025sinkfreeorientationslocalsampler}{article}{
      author={Anand, Konrad},
      author={Freifeld, Graham},
      author={Guo, Heng},
      author={Wang, Chunyang},
      author={Wang, Jiaheng},
       title={Sink-free orientations: a local sampler with applications},
        date={2025},
      eprint={arXiv.2502.05877},
         url={https://arxiv.org/abs/2502.05877},
}

\bib{Barbook}{book}{
      author={Barvinok, Alexander},
       title={Combinatorics and complexity of partition functions},
      series={Algorithms and combinatorics},
   publisher={Springer},
        date={2016},
      volume={30},
        ISBN={978-3-319-51828-2},
         url={https://doi.org/10.1007/978-3-319-51829-9},
}

\bib{bencs2024deterministicapproximatecountingcolorings}{article}{
      author={Bencs, Ferenc},
      author={Berrekkal, Khallil},
      author={Regts, Guus},
       title={Deterministic approximate counting of colorings with fewer than
  $2{\Delta}$ colors via absence of zeros},
        date={2024},
      eprint={arXiv.2408.04727},
         url={https://arxiv.org/abs/2408.04727},
}

\bib{Borgsetal}{article}{
      author={Borgs, Christian},
      author={Chayes, Jennifer},
      author={Kahn, Jeff},
      author={Lov\'asz, L\'aszl\'o},
       title={Left and right convergence of graphs with bounded degree},
        date={2013},
        ISSN={1042-9832,1098-2418},
     journal={Random Structures Algorithms},
      volume={42},
      number={1},
       pages={1\ndash 28},
         url={https://doi.org/10.1002/rsa.20414},
      review={\MR{2999210}},
}

\bib{bencs2020some}{article}{
      author={Bencs, Ferenc},
      author={Csikv\'{a}ri, P\'{e}ter},
      author={Regts, Guus},
       title={Some applications of {W}agner's weighted subgraph counting
  polynomial},
        date={2021},
     journal={Electron. J. Combin.},
      volume={28},
      number={4},
       pages={Paper No. 4.14, 21},
         url={https://doi-org.proxy.uba.uva.nl/10.37236/10185},
      review={\MR{4328912}},
}

\bib{BanGar}{article}{
      author={Bandyopadhyay, Antar},
      author={Gamarnik, David},
       title={Counting without sampling: asymptotics of the log-partition
  function for certain statistical physics models},
        date={2008},
        ISSN={1042-9832},
     journal={Random Structures Algorithms},
      volume={33},
      number={4},
       pages={452\ndash 479},
         url={https://doi-org.proxy.uba.uva.nl/10.1002/rsa.20236},
      review={\MR{2462251}},
}

\bib{Bayatietalmatchings}{inproceedings}{
      author={Bayati, Mohsen},
      author={Gamarnik, David},
      author={Katz, Dimitriy},
      author={Nair, Chandra},
      author={Tetali, Prasad},
       title={Simple deterministic approximation algorithms for counting
  matchings},
        date={2007},
   booktitle={S{TOC}'07---{P}roceedings of the 39th {A}nnual {ACM} {S}ymposium
  on {T}heory of {C}omputing},
   publisher={ACM, New York},
       pages={122\ndash 127},
         url={https://doi-org.proxy.uba.uva.nl/10.1145/1250790.1250809},
      review={\MR{2402435}},
}

\bib{bencs2024approximatingvolumetruncatedrelaxation}{article}{
      author={Bencs, Ferenc},
      author={Regts, Guus},
       title={Approximating the volume of a truncated relaxation of the
  independence polytope},
        date={2024},
      eprint={arXiv.2404.08577},
         url={https://arxiv.org/abs/2404.08577},
}

\bib{bencs2025improvedboundszeroschromatic}{article}{
      author={Bencs, Ferenc},
      author={Regts, Guus},
       title={Improved bounds on the zeros of the chromatic polynomial of
  graphs and claw-free graphs},
        date={2025},
      eprint={arXiv.2505.04366},
         url={https://arxiv.org/abs/2505.04366},
}

\bib{BarvinokSoberon}{article}{
      author={Barvinok, Alexander},
      author={Sober\'{o}n, Pablo},
       title={Computing the partition function for graph homomorphisms},
        date={2017},
        ISSN={0209-9683},
     journal={Combinatorica},
      volume={37},
      number={4},
       pages={633\ndash 650},
         url={https://doi-org.proxy.uba.uva.nl/10.1007/s00493-016-3357-2},
      review={\MR{3694706}},
}

\bib{chen2025deterministiccountingcouplingindependence}{article}{
      author={Chen, Xiaoyu},
      author={Feng, Weiming},
      author={Guo, Heng},
      author={Zhang, Xinyuan},
      author={Zou, Zongrui},
       title={Deterministic counting from coupling independence},
        date={2025},
      eprint={arXiv.2410.23225},
         url={https://arxiv.org/abs/2410.23225},
}

\bib{deboerratios}{article}{
      author={de~Boer, David},
      author={Buys, Pjotr},
      author={Guerini, Lorenzo},
      author={Peters, Han},
      author={Regts, Guus},
       title={Zeros, chaotic ratios and the computational complexity of
  approximating the independence polynomial},
        date={2024},
        ISSN={0305-0041,1469-8064},
     journal={Math. Proc. Cambridge Philos. Soc.},
      volume={176},
      number={2},
       pages={459\ndash 494},
         url={https://doi.org/10.1017/s030500412300052x},
      review={\MR{4706780}},
}

\bib{FerProc}{article}{
      author={Fern\'{a}ndez, Roberto},
      author={Procacci, Aldo},
       title={Regions without complex zeros for chromatic polynomials on graphs
  with bounded degree},
        date={2008},
        ISSN={0963-5483},
     journal={Combin. Probab. Comput.},
      volume={17},
      number={2},
       pages={225\ndash 238},
         url={https://doi-org.proxy.uba.uva.nl/10.1017/S0963548307008632},
      review={\MR{2396349}},
}

\bib{gamarnik2022correlation}{article}{
      author={Gamarnik, David},
       title={Correlation decay and the absence of zeros property of partition
  functions},
        date={2023},
        ISSN={1042-9832,1098-2418},
     journal={Random Structures Algorithms},
      volume={62},
      number={1},
       pages={155\ndash 180},
         url={https://doi.org/10.1002/rsa.21083},
      review={\MR{4516081}},
}

\bib{BeyondlambdacGalanisetal}{article}{
      author={Galanis, Andreas},
      author={\v{S}tefankovi\v{c}, Daniel},
      author={Vigoda, Eric},
       title={Inapproximability of the partition function for the
  antiferromagnetic {I}sing and hard-core models},
        date={2016},
        ISSN={0963-5483},
     journal={Combin. Probab. Comput.},
      volume={25},
      number={4},
       pages={500\ndash 559},
         url={https://doi-org.proxy.uba.uva.nl/10.1017/S0963548315000401},
      review={\MR{3506425}},
}

\bib{HPR21}{article}{
      author={Helmuth, Tyler},
      author={Perkins, Will},
      author={Regts, Guus},
       title={Algorithmic {P}irogov-{S}inai theory},
        date={2020},
        ISSN={0178-8051},
     journal={Probab. Theory Related Fields},
      volume={176},
      number={3-4},
       pages={851\ndash 895},
         url={https://doi-org.proxy.uba.uva.nl/10.1007/s00440-019-00928-y},
      review={\MR{4087485}},
}

\bib{JPR2024}{article}{
      author={Jenssen, Matthew},
      author={Patel, Viresh},
      author={Regts, Guus},
       title={Improved bounds for the zeros of the chromatic polynomial via
  {W}hitney's broken circuit theorem},
        date={2024},
        ISSN={0095-8956,1096-0902},
     journal={J. Combin. Theory Ser. B},
      volume={169},
       pages={233\ndash 252},
         url={https://doi.org/10.1016/j.jctb.2024.06.005},
      review={\MR{4773549}},
}

\bib{edge}{inproceedings}{
      author={Lin, Chengyu},
      author={Liu, Jingcheng},
      author={Lu, Pinyan},
       title={A simple {FPTAS} for counting edge covers},
        date={2014},
   booktitle={Proceedings of the {T}wenty-{F}ifth {A}nnual {ACM}-{SIAM}
  {S}ymposium on {D}iscrete {A}lgorithms},
   publisher={ACM, New York},
       pages={341\ndash 348},
         url={https://doi-org.proxy.uba.uva.nl/10.1137/1.9781611973402.25},
      review={\MR{3376385}},
}

\bib{weighted}{incollection}{
      author={Liu, Jingcheng},
      author={Lu, Pinyan},
      author={Zhang, Chihao},
       title={F{PTAS} for counting weighted edge covers},
        date={2014},
   booktitle={Algorithms---{ESA} 2014},
      series={Lecture Notes in Comput. Sci.},
      volume={8737},
   publisher={Springer, Heidelberg},
       pages={654\ndash 665},
         url={https://doi-org.proxy.uba.uva.nl/10.1007/978-3-662-44777-2_54},
      review={\MR{3253169}},
}

\bib{LSS2Delta}{incollection}{
      author={Liu, Jingcheng},
      author={Sinclair, Alistair},
      author={Srivastava, Piyush},
       title={A deterministic algorithm for counting colorings with {$2\Delta$}
  colors},
        date={2019},
   booktitle={2019 {IEEE} 60th {A}nnual {S}ymposium on {F}oundations of
  {C}omputer {S}cience},
   publisher={IEEE Comput. Soc. Press, Los Alamitos, CA},
       pages={1380\ndash 1404},
         url={https://doi-org.proxy.uba.uva.nl/10.1109/FOCS.2019.00085},
      review={\MR{4228232}},
}

\bib{LSSfisher}{article}{
      author={Liu, Jingcheng},
      author={Sinclair, Alistair},
      author={Srivastava, Piyush},
       title={Fisher zeros and correlation decay in the {I}sing model},
        date={2019},
        ISSN={0022-2488},
     journal={J. Math. Phys.},
      volume={60},
      number={10},
       pages={103304, 12},
         url={https://doi-org.proxy.uba.uva.nl/10.1063/1.5082552},
      review={\MR{4021065}},
}

\bib{LSScorrelation}{article}{
      author={Liu, Jingcheng},
      author={Sinclair, Alistair},
      author={Srivastava, Piyush},
       title={Correlation decay and partition function zeros: Algorithms and
  phase transitions},
        date={2022},
     journal={SIAM Journal on Computing},
      volume={0},
      number={0},
       pages={FOCS19\ndash 200},
}

\bib{lu2013improved}{incollection}{
      author={Lu, Pinyan},
      author={Yin, Yitong},
       title={Improved fptas for multi-spin systems},
        date={2013},
   booktitle={Approximation, randomization, and combinatorial optimization.
  algorithms and techniques},
   publisher={Springer},
       pages={639\ndash 654},
}

\bib{Moitra}{article}{
      author={Moitra, Ankur},
       title={Approximate counting, the {L}ov\'asz local lemma, and inference
  in graphical models},
        date={2019},
        ISSN={0004-5411,1557-735X},
     journal={J. ACM},
      volume={66},
      number={2},
       pages={Art. 10, 25},
         url={https://doi.org/10.1145/3268930},
      review={\MR{3936375}},
}

\bib{penrose}{article}{
      author={Penrose, Oliver},
       title={Convergence of fugacity expansions for classical systems},
        date={1967},
     journal={Statistical mechanics: foundations and applications},
       pages={101},
}

\bib{PatReg17}{article}{
      author={Patel, Viresh},
      author={Regts, Guus},
       title={Deterministic polynomial-time approximation algorithms for
  partition functions and graph polynomials},
        date={2017},
     journal={{SIAM} J. Comput.},
      volume={46},
      number={6},
       pages={1893\ndash 1919},
         url={https://doi.org/10.1137/16M1101003},
}

\bib{PRSokal}{article}{
      author={Peters, Han},
      author={Regts, Guus},
       title={On a conjecture of {S}okal concerning roots of the independence
  polynomial},
        date={2019},
        ISSN={0026-2285},
     journal={Michigan Math. J.},
      volume={68},
      number={1},
       pages={33\ndash 55},
         url={https://doi-org.proxy.uba.uva.nl/10.1307/mmj/1541667626},
      review={\MR{3934603}},
}

\bib{regtsabsence}{article}{
      author={Regts, Guus},
       title={Absence of zeros implies strong spatial mixing},
        date={2023},
        ISSN={0178-8051,1432-2064},
     journal={Probab. Theory Related Fields},
      volume={186},
      number={1-2},
       pages={621\ndash 641},
         url={https://doi.org/10.1007/s00440-023-01190-z},
      review={\MR{4586228}},
}

\bib{Sokalzeros}{article}{
      author={Sokal, Alan~D.},
       title={Bounds on the complex zeros of (di)chromatic polynomials and
  {P}otts-model partition functions},
        date={2001},
        ISSN={0963-5483},
     journal={Combin. Probab. Comput.},
      volume={10},
      number={1},
       pages={41\ndash 77},
         url={https://doi-org.proxy.uba.uva.nl/10.1017/S0963548300004612},
      review={\MR{1827809}},
}

\bib{ScottSokal}{article}{
      author={Scott, Alexander~D.},
      author={Sokal, Alan~D.},
       title={The repulsive lattice gas, the independent-set polynomial, and
  the {L}ov\'{a}sz local lemma},
        date={2005},
        ISSN={0022-4715},
     journal={J. Stat. Phys.},
      volume={118},
      number={5-6},
       pages={1151\ndash 1261},
         url={https://doi-org.proxy.uba.uva.nl/10.1007/s10955-004-2055-4},
      review={\MR{2130890}},
}

\bib{BeyondlambdacSlyandSun}{article}{
      author={Sly, Allan},
      author={Sun, Nike},
       title={Counting in two-spin models on {$d$}-regular graphs},
        date={2014},
        ISSN={0091-1798},
     journal={Ann. Probab.},
      volume={42},
      number={6},
       pages={2383\ndash 2416},
         url={https://doi-org.proxy.uba.uva.nl/10.1214/13-AOP888},
      review={\MR{3265170}},
}

\bib{Shaounified}{article}{
      author={Shao, Shuai},
      author={Sun, Yuxin},
       title={Contraction: a unified perspective of correlation decay and
  zero-freeness of 2-spin systems},
        date={2021},
        ISSN={0022-4715,1572-9613},
     journal={J. Stat. Phys.},
      volume={185},
      number={2},
       pages={Paper No. 12, 25},
         url={https://doi.org/10.1007/s10955-021-02831-0},
      review={\MR{4332156}},
}

\bib{shao2025zero}{article}{
      author={Shao, Shuai},
      author={Shi, Ke},
       title={Zero-freeness is all you need: A {W}eitz-type {FPTAS} for the
  entire {L}ee-{Y}ang zero-free region},
        date={2025},
     journal={arXiv preprint arXiv:2509.06623},
}

\bib{shao2024zero}{article}{
      author={Shao, Shuai},
      author={Ye, Xiaowei},
       title={From zero-freeness to strong spatial mixing via a
  {C}hristoffel-{D}arboux type identity},
        date={2024},
     journal={arXiv preprint arXiv:2401.09317},
}

\bib{Weitz}{inproceedings}{
      author={Weitz, Dror},
       title={Counting independent sets up to the tree threshold},
        date={2006},
   booktitle={S{TOC}'06: {P}roceedings of the 38th {A}nnual {ACM} {S}ymposium
  on {T}heory of {C}omputing},
   publisher={ACM, New York},
       pages={140\ndash 149},
         url={https://doi-org.proxy.uba.uva.nl/10.1145/1132516.1132538},
      review={\MR{2277139}},
}

\bib{Whitney}{article}{
      author={Whitney, Hassler},
       title={A logical expansion in mathematics},
        date={1932},
        ISSN={0002-9904},
     journal={Bull. Amer. Math. Soc.},
      volume={38},
      number={8},
       pages={572\ndash 579},
         url={https://doi.org/10.1090/S0002-9904-1932-05460-X},
      review={\MR{1562461}},
}

\end{biblist}
\end{bibdiv}
\end{document}